\documentclass[a4paper]{article}

\usepackage{amsmath}
\usepackage{amssymb}
\usepackage{graphicx}
\usepackage[OT4]{fontenc}
\usepackage[latin2]{inputenc}
\usepackage{url}
\usepackage[numbers]{natbib}

\newcommand{\okra}[1]{\left( #1 \right)}
\newcommand{\kwad}[1]{\left[ #1 \right]}
\newcommand{\klam}[1]{\left\{ #1 \right\}}

\newcommand{\boole}[1]{{\bf 1}{\klam{#1}}}

\DeclareMathOperator{\card}{\#}
\DeclareMathOperator{\sred}{\mathbf{E}}

\DeclareMathOperator{\PPM}{PPM}
\DeclareMathOperator{\LZ}{LZ}

\DeclareMathOperator*{\hilberg}{hilb}

\newtheorem{definition}{Definition}

\newtheorem{theorem}{Theorem}

\newenvironment*{proof}{\begin{trivlist}\item[]
\noindent\textbf{Proof:}}{$\Box$\par\end{trivlist}}
\newenvironment*{proof*}[1]{\begin{trivlist}\item[]
\noindent\textbf{Proof of #1:}}{$\Box$\par\end{trivlist}}

\begin{document}

\begin{titlepage}

  \title{On a Class of Markov Order Estimators \\ Based on
    PPM and Other Universal Codes}

\author{{\L}ukasz D\k{e}bowski\thanks{
    {\L}. D\k{e}bowski is with
    the Institute of Computer Science, Polish Academy of Sciences, 
    ul. Jana Kazimierza 5, 01-248 Warszawa, Poland 
    (e-mail: ldebowsk@ipipan.waw.pl).}}
\date{}

\maketitle

\begin{abstract}
  We investigate a class of estimators of the Markov order for
  stationary ergodic processes which form a slight modification of the
  constructions by Merhav, Gutman, and Ziv in 1989 as well as by Ryabko,
  Astola, and Malyutov in 2006 and 2016. All the considered estimators
  compare the estimate of the entropy rate given by a universal code
  with the empirical conditional entropy of a string and return the
  order for which the two quantities are approximately equal.
  However, our modification, which we call universal Markov orders,
  satisfies a few attractive properties, not shown by the mentioned
  authors for their original constructions.  Firstly, the universal
  Markov orders are almost surely consistent, without any
  restrictions. Secondly, they are upper bounded asymptotically by the
  logarithm of the string length divided by the entropy rate. Thirdly,
  if we choose the Prediction by Partial Matching (PPM) as the
  universal code then the number of distinct substrings of the length
  equal to the universal Markov order constitutes an upper bound for
  the block mutual information.  Thus universal Markov orders can be
  also used indirectly for quantification of long memory for an
  ergodic process.
  \\[2ex]
\end{abstract}


\end{titlepage}
\pagestyle{plain}   


\section{Introduction}
\label{secIntro}

Throughout this paper, we denote sequences $x_j^k=(x_i)_{j\le i\le k}$
over a finite alphabet $\mathbb{X}=\klam{a_1,a_2,...,a_D}$, where
$D\ge 2$ and $x_j^{j-1}=\lambda$ equals the empty string.  For a
stationary probability measure $P$ on infinite sequences over alphabet
$\mathbb{X}$ and random variables $X_k(x_1^\infty):=x_k$, we use the
abridged notation $P(x_1^n):=P(X_1^n=x_1^n)$ and
$P(x_j^n|x_1^{j-1}):=P(X_j^n=x_j^n|X_1^{j-1}=x_1^{j-1})$.  For the
stationary measure $P$ as above, we also define the Markov order
\begin{align}
  M^P:=\inf\klam{k\ge 0: P(x_{k+1}^n|x_1^k)
  =\prod_{i=k+1}^n P(x_i|x_{i-k}^{i-1})
  \text{ for all strings $x_1^n$}},
\end{align}
where the infimum of the empty set equals infinity,
$\inf\emptyset:=\infty$.  If the Markov order $M^P=M$ is finite,
measure $P$ is called an $M$-th order Markov measure.

Several estimators of the Markov order for stationary ergodic measures
were exhibited in the literature
\cite{MerhavGutmanZiv89,CsiszarShields00,Csiszar02,MorvaiWeiss05,RyabkoAstola06,
  CsiszarTalata06,Talata13,BaigorriGocalvesResende14,RyabkoAstolaMalyutov16,
  PapapetrouKugiumtzis16}.  The goal of this paper is to investigate a
class of simple consistent estimators of the Markov order based on
universal codes---such as the prefix-free Kolmogorov complexity
\cite{Chaitin75,LiVitanyi08,Brudno82}, the Lempel-Ziv code
\cite{ZivLempel77}, Prediction by Partial Matching (PPM)
\cite{Ryabko84en2,ClearyWitten84}, or one of many grammar-based codes
\cite{KiefferYang00,CharikarOthers05,Debowski11b}.  In fact, our
estimators are extremely close to the idea initially proposed by
Merhav, Gutman, and Ziv \cite{MerhavGutmanZiv89} and by Ryabko, Astola,
and Malyutov \cite{RyabkoAstola06,RyabkoAstolaMalyutov16} but,
tampering with fine definition details, we are able to prove easily
their almost sure consistency and a few other neat properties which we
could not find in
\cite{MerhavGutmanZiv89,RyabkoAstola06,RyabkoAstolaMalyutov16}.

As we have mentioned, the Markov order estimators constructed here
are parameterized by universal codes. Simply speaking, our estimators
compare the estimate of the entropy rate given by a universal code
with the empirical conditional entropy of a string and return the
order of the empirical conditional entropy for which the two
quantities are roughly equal. Thus, the difference between our
estimators and the constructions of
\cite{MerhavGutmanZiv89,RyabkoAstola06,RyabkoAstolaMalyutov16} is
quite fine.  Let $h_k(x_1^n)$ be the empirical conditional entropy of
order $k$, to be formally defined in Section \ref{secOrder}, and let
$\LZ(x_1^n)$ be the length of the Lempel-Ziv code
\cite{ZivLempel77}. In Equation (14) of \cite{MerhavGutmanZiv89}
adjusted to our notation the Markov order estimator is defined as
\begin{align}
  \label{LZOrderMGZ}
  \mathbf{M}_\lambda(x_1^n):=\min\klam{k\ge 0: h_k(x_1^n)\le
  \frac{1}{n}\LZ(x_1^n)+\lambda}
\end{align}
with a parameter $\lambda>0$ fixed. Moreover, in
\cite{MerhavGutmanZiv89}, we can find some bounds for the error
probability $P(\mathbf{M}_\lambda(X_1^n)\neq k)$ under the hypothesis
that $M^P=k<\infty$. In contrast, in
\cite{RyabkoAstola06,RyabkoAstolaMalyutov16}, the authors investigated
testing the null hypothesis that the probability measure has a Markov
order $M^P\le M<\infty$ versus the alternative $M^P> M$.  For that
goal, they proposed the critical region defined by inequality
\begin{align}
  \label{LZOrderRAM}
  (n-M)h_M(x_1^n)\le
  \LZ(x_1^n)+\log(1/\alpha),
\end{align}
where $\alpha\in(0,1)$. It turned out that the type I error
probability is upper bounded by $\alpha$, whereas the type II error
probability tends to $0$ for $n$ going to infinity.

In contrast, in this paper we are interested in the almost surely
consistent estimation rather than in hypothesis testing. A respective
example of our Markov order estimators is
\begin{align}
  \label{LZOrder}
  \mathbf{M}(x_1^n):=\min\klam{k\ge 0: (n-k)h_k(x_1^n)\le
  \LZ(x_1^n)+\log\frac{\pi^2}{6}+2\log (n+1)}.
\end{align}
In contrast to the original ideas of
\cite{MerhavGutmanZiv89,RyabkoAstola06,RyabkoAstolaMalyutov16}, this
cosmetic change allows to demonstrate almost sure consistency of the
Markov order estimator (\ref{LZOrder}) without much deliberation.  In
particular, we can substitute the length $\LZ(x_1^n)$ with the length
of any universal code to obtain a whole class of strongly consistent
estimators. Whereas the idea of using an arbitrary universal code to
define critical region (\ref{LZOrderRAM}) was already postulated and
proved true in \cite{RyabkoAstola06,RyabkoAstolaMalyutov16}, we push
it somewhat further to obtain a more elegant theory.  Since the
essential ideas of Markov order estimation based on universal codes
have been proposed by
\cite{MerhavGutmanZiv89,RyabkoAstola06,RyabkoAstolaMalyutov16}, our
merit lies mostly in the aesthetics of construction and proving a few
more new properties besides the consistency.  We deem that we give a
final touch to propositions that may have circulated in the folklore.

We call the particular estimators introduced here universal Markov
orders of a string.  In the following, we will show that universal
Markov orders enjoy several nice properties. To be concrete, in the
subsequent three sections, using elementary methods, we will
demonstrate the following results:
\begin{itemize}
\item In Section~\ref{secOrder}, we will first define the necessary
  concepts to introduce a formal definition of universal Markov
  orders, generalizing the definition of estimator
  (\ref{LZOrder}). Secondly, we will prove that universal Markov
  orders are almost surely consistent---for any universal code and for
  any stationary ergodic measure.  Since we will demonstrate the
  general consistency of universal Markov orders using the Barron
  lemma \cite[Theorem 3.1]{Barron85b}, probably the length of the
  universal code in (\ref{LZOrder}) cannot be substituted with the
  consistent estimators of entropy rate introduced in
  \cite{KontoyiannisOthers98,GaoKontoyiannisBienenstock08}.
\item Section~\ref{secBounds} is devoted to demonstrating that
  universal Markov orders are upper bounded asymptotically by
  $h^{-1}\log n$ almost surely, $h$ being the entropy rate. This
  result will be contrasted with a naive upper bound given by the
  maximal repetition length, which is asymptotically greater than
  $h^{-1}\log n$ almost surely and whose behavior is better captured
  by the R\'enyi entropy rates
  \cite{Shields92b,Szpankowski93a,Debowski18b}.
\item Finally, let the vocabulary size of a given order for a
  string---also called the subword complexity \cite{DeLuca99}---be the
  number of distinct substrings of the length equal to the order.
  According to Section~\ref{secPPM}, if we choose the PPM as the
  universal code then the vocabulary size of the universal Markov
  order constitutes an upper bound for the block mutual
  information. This result strengthens some similar propositions of
  \cite{Debowski18} which find their applications in statistical
  language modeling \cite{Debowski11b,Debowski12,Debowski17}.  The
  analogical results of \cite{Debowski18} involve the vocabulary size
  of a larger order, equal to the Krichevsky-Trofimov Markov order
  estimator, proved to be inconsistent for the uniform measure by
  Csiszar and Shields \cite{CsiszarShields00}. In contrast, the
  universal Markov order applied in Section~\ref{secPPM} is a
  consistent estimator as shown in Section~\ref{secOrder} and is
  usually negligible compared to the vocabulary size because of the
  bound developed in Section~\ref{secBounds}. This remains in contrast
  with the unknown worst-case behavior of the Krichevsky-Trofimov
  estimator.
\end{itemize}

In this way, universal Markov orders can be successfully applied not
only to consistent estimation of the Markov order for arbitrary
stationary ergodic measures but also to other problems connected with
quantification of long memory, compare \cite{Talata13}.  We hope that
these estimators may inspire further constructions. Especially
interesting seems the extension to countably infinite alphabets, for
which there are known consistent Markov order estimators within the
class of Markov measures \cite{MorvaiWeiss05} but there are more
general problems with the existence of universal codes
\cite{Kieffer78,OrlitskySanthanamZhang04,SilvaPiantanida20}.

\section{Universal Markov orders}
\label{secOrder}

In this section, we will define universal Markov orders and we will
establish their consistency. First, to fix the notation, let us recall
the concepts of the empirical vocabulary, the empirical entropies, the
true entropies, and universal coding. Let
$\mathbb{X}^*=\bigcup_{n\ge 0}\mathbb{X}^n$ denote the set of strings
of an arbitrary length including the empty string
$\lambda\in \mathbb{X}^0$. We define the frequency of a substring
$w_1^k\in\mathbb{X}^k$ in a string $x_1^n\in\mathbb{X}^n$ where
$0\le k\le n$ as
\begin{align}
    N(w_1^k|x_1^n):=\sum_{i=1}^{n-k+1}\boole{x_i^{i+k-1}=w_1^k}.
\end{align}
We use this notation also for the empty string $w_1^0=\lambda$, where
$N(\lambda|x_1^n)=n+1$, according to the above definition. Subsequently, we define the empirical
vocabulary of a string $x_1^n$ of order $k$ as
\begin{align}
  V_k(x_1^n):=\klam{w_1^k\in\mathbb{X}^k: N(w_1^k|x_1^n)>0}.
\end{align}
We use this notation also for $k=0$, where
$V_0(x_1^n)=\klam{\lambda}$, according to the above definition.  Let
$\card A$ denote the cardinality of set $A$. The vocabulary size
$\card V_k(x_1^m)$ is also called the subword complexity
\cite{DeLuca99}. Using the empirical vocabulary, the empirical
(conditional) entropy of a string $x_1^n$ of order $k\ge 0$ can be
equivalently defined as
\begin{align}
  h_k(x_1^n)
  &:=
  \sum_{w_1^{k+1}\in V_{k+1}(x_1^{n})}
  \frac{N(w_1^{k+1}|x_1^n)}{n-k}
    \log \frac{N(w_1^k|x_1^{n-1})}{N(w_1^{k+1}|x_1^n)}
    \nonumber\\
  &=
    \frac{1}{n-k}
    \sum_{i=k+1}^n \log \frac{N(x_{i-k}^{i-1}|x_1^{n-1})}{N(x_{i-k}^{i}|x_1^n)}
  ,
\end{align}
where $\log x$ stands for the binary logarithm of $x$.

We notice monotonicity of the empirical entropy.
\begin{theorem}
  \label{theoHGrowing}
  We have $0\le h_k(x_2^n)-h_{k+1}(x_1^n)\le \log D$.
\end{theorem}
\begin{proof}
  Quantity $h_k(x_2^n)-h_{k+1}(x_1^n)$ is the conditional mutual
  information between two random variables, where each assumes $D$
  distinct values.
\end{proof}
\begin{theorem}
  \label{theoHGrowingII}
  We have
  $0\le h_k(x_1^n)-\frac{n-1-k}{n-k}h_k(x_2^n)\le \log
  \min\klam{2,D}$.
\end{theorem}
\begin{proof}
  Quantity
  $h_k(x_1^n)-\frac{n-1-k}{n-k}h_k(x_2^n)-\frac{1}{n-k}h_k(x_1^{k+1})$
  is the conditional mutual information between two random variables,
  one of which assumes two distinct values and another assumes $D$
  distinct values---compare with \cite[Theorem A6]{Debowski18}, which
  only showed the left inequality. To obtain the claim we notice that
  $h_k(x_1^{k+1})=0$.
\end{proof}
Thus, by Theorems \ref{theoHGrowing} and \ref{theoHGrowingII},
quantity $(n-k)h_k(x_1^n)$ decreases with $k$ since
\begin{align}
  (n-k)h_k(x_1^n)\ge (n-k-1)h_k(x_2^n)\ge (n-k-1)h_{k+1}(x_1^n).
\end{align}

Extending the observation made in Theorem \ref{theoHGrowingII}, we can
also demonstrate superadditivity of the empirical entropy, which will be
used in Section \ref{secPPM}.
\begin{theorem}
  \label{theoHSuperadditive}
  For $0\le k<n,m-n<m$, we have inequality
  \begin{align}
    \label{HSuperadditive}
    0\le
    h_k(x_1^m)
    -\frac{n-k}{m-k}h_k(x_1^n)
    -\frac{k}{m-k}h_k(x_{n-k}^{n+k})
    -\frac{m-n-k}{m-k}h_k(x_{n+1}^m)
    \le C
  \end{align}  
where $C=\log \min\klam{3,D}$.
\end{theorem}
\begin{proof}
  The sandwich-bounded quantity is the conditional mutual information
  between two random variables, one of which assumes three distinct
  values and another assumes $D$ distinct values, see \cite[Theorem
  A6]{Debowski18}.
\end{proof}

The empirical entropies will be now contrasted with the true
conditional entropies.  For a stationary probability measure $P$ on
infinite sequences over alphabet $\mathbb{X}$, let us introduce the
conditional entropies $h_k^P$ and the entropy rate $h^P$ defined as
\begin{align}
  h_k^P&:=\sred\kwad{-\log P(X_i|X_{i-k}^{i-1})}
  ,
  \\
  h^P&:=\inf_{k\in\mathbb{N}} h_k^P
  =\lim_{n\to\infty}\frac{1}{n}\sred\kwad{-\log P(X_1^n)}
  .
\end{align}
In the following assume that $P$ is additionally ergodic. Then
by the Birkhoff ergodic theorem \cite{Garsia65}, we have
\begin{align}
  \label{BirkhoffHk}
  \lim_{n\to\infty} h_k(X_1^n)=\lim_{n\to\infty}
  \frac{1}{n}\kwad{-\log \prod_{i=k+1}^n P(X_i|X_{i-k}^{i-1})}=
  h_k^P
  \text{ almost surely},
\end{align}
whereas the Shannon-McMillan-Breiman theorem \cite{AlgoetCover88}
yields
\begin{align}
  \lim_{n\to\infty}\frac{1}{n}\kwad{-\log P(X_1^n)}=h^P
  \text{ almost surely}.
\end{align}

Subsequently, let us approach the problem of universal coding from a
more abstract perspective inspired by the concept of algorithmic
probability in algorithmic information theory \cite{LiVitanyi08}.  A
semi-distribution is a real function $\mathbf{\Pi}$ of strings
$w\in\mathbb{X}^*=\bigcup_{n\ge 0}\mathbb{X}^n$ of an arbitrary length
(including the empty string) such that $\mathbf{\Pi}(w)\ge 0$ and the
Kraft inequality $\sum_{w\in\mathbb{X}^*}\mathbf{\Pi}(w)\le 1$ is
satisfied. Let us write the respective pointwise entropy as
$\mathbf{H}(x_1^n):=-\log \mathbf{\Pi}(x_1^n)$.  The semi-distribution
$\mathbf{\Pi}$ is called universal if for any stationary ergodic
probability measure $P$ on infinite sequences over a finite alphabet
$\mathbb{X}$ we have
\begin{align}
  \label{Universal}
  \lim_{n\to\infty}\frac{\mathbf{H}(X_1^n)}{n}=h^P
  \text{ almost surely}.
\end{align}
Examples of universal semi-distributions are well known. In
particular, the pointwise entropy $\mathbf{H}(x_1^n)$ can be chosen as
the prefix-free Kolmogorov complexity of string $x_1^n$
\cite{Chaitin75,LiVitanyi08,Brudno82}, which is uncomputable. In this
case, $\mathbf{\Pi}(x_1^n)$ equals approximately the algorithmic
probability of $x_1^n$---by the coding theorem
\cite{Chaitin75,LiVitanyi08}. A more feasible choice of
$\mathbf{H}(x_1^n)$ is the length of any computable universal code,
such as the Lempel-Ziv code \cite{ZivLempel77}, the PPM code
\cite{Ryabko84en2,ClearyWitten84}, or one of many grammar-based codes
\cite{KiefferYang00,CharikarOthers05,Debowski11b}---where for the
universal codes we need a length correction such as
$\log\frac{\pi^2}{6}+2\log (n+1)$ in equation (\ref{LZOrder}). We
stress that the prefix-free Kolmogorov complexity does not require
this correction since its Kraft sum equals the halting probability
$\mathbf{\Omega}$, strictly less than one.

Now we have everything to define the main concept of this paper, i.e.,
the universal Markov orders---slightly modified with respect to the
constructions in \cite{MerhavGutmanZiv89,RyabkoAstola06,RyabkoAstolaMalyutov16}.
\begin{definition}
  \label{defiMOrder}
  Let $\mathbf{\Pi}$ be a universal semi-distribution and let
  $\mathbf{H}(x_1^n):=-\log \mathbf{\Pi}(x_1^n)$ be the respective
  pointwise entropy.  The respective universal Markov order of a
  string $x_1^n$ is defined as
  \begin{align}
    \label{MOrder}
    \mathbf{M}(x_1^n):=\min\klam{k\ge 0: (n-k)h_k(x_1^n)\le \mathbf{H}(x_1^n)}.
\end{align}
\end{definition}

Subsequently, we observe that the universal Markov order of a string
is a consistent estimator of the Markov order of a stationary ergodic
probability measure. This proposition complements and strengthens the
results of
\cite{MerhavGutmanZiv89,RyabkoAstola06,RyabkoAstolaMalyutov16}
discussed in Section \ref{secIntro}.
\begin{theorem}
  \label{theoConsistency}
  For a stationary ergodic probability measure $P$ on infinite
  sequences over a finite alphabet and a universal Markov
  order $\mathbf{M}$, we have
  \begin{align}
    \label{Consistency}
    \lim_{n\to\infty} \mathbf{M}(X_1^n)&=M^P \text{ almost surely}.
  \end{align}
\end{theorem}
\begin{proof}
  Suppose first that $M^P=M$ is finite.  By the non-negativity of the
  Kullback-Leibler divergence between the empirical distribution and
  the $k$-th order Markov approximation of a stationary measure $P$, we have
  \begin{align}
    (n-k)h_k(x_1^n)\le -\log \prod_{i=k+1}^n P(x_i|x_{i-k}^{i-1}).
  \end{align}
  Putting $k=M$, we obtain
  \begin{align}
    (n-M)h_M(x_1^n)\le -\log P(x_{M+1}^n|x_1^M)\le -\log P(x_1^n).
  \end{align}
  On the other hand, by the Barron lemma \cite[Theorem 3.1]{Barron85b}
  for any semi-distribution $\mathbf{\Pi}$ and the respective pointwise
  entropy, we have
  \begin{align}
    \label{Barron}
    \lim_{n\to\infty} \kwad{\mathbf{H}(X_1^n)+\log P(X_1^n)}=\infty
    \text{ almost surely}.
  \end{align}
  Hence  almost surely, for sufficiently large $n$, we obtain
  \begin{align}
    (n-M)h_M(X_1^n)\le \mathbf{H}(X_1^n).
  \end{align}
  In other words, for these $n$, we have $\mathbf{M}(X_1^n)\le M$.

  Subsequently, assume an arbitrary $M^P$.  By the definition of the
  Markov order we have $h_k^P>h^P$ for $k<M^P$.  Recall that we have
  (\ref{BirkhoffHk}) and (\ref{Universal}).  Hence, almost surely, for
  each $k<M^P$ and all sufficiently large $n$, we obtain
  $(n-k)h_k(X_1^n)>\mathbf{H}(X_1^n)$. Thus $\mathbf{M}(X_1^n)>k$
  since $(n-k)h_k(X_1^n)$ is a decreasing function of $k$. Combining
  this result with the observation made in the previous paragraph
  yields the claim.
\end{proof}

Paying another tribute to the algorithmic information theory, the
Barron-lemma-like property (\ref{Barron}) for $\mathbf{H}(x_1^n)$
being the prefix-free Kolmogorov complexity is well known to
characterize equivalently, via the Schnorr theorem \cite{LiVitanyi08},
the set of Martin-L\"of random sequences. Of course, the probability
of this set equals one. Using the effective Birkhoff ergodic theorem
\cite{FranklinOthers12,BienvenuOthers12}, it can be shown in fact that
convergence (\ref{Consistency}) holds on all Martin-L\"of random
sequences for any pointwise entropy function $\mathbf{H}(x_1^n)$
greater than the prefix-free Kolmogorov complexity if the universality
condition (\ref{Universal}) also holds on all Martin-L\"of random
sequences. That is, consistency of universal Markov orders for
computable universal codes can be easily restated as a so called
effective law of probability---if these codes do not do something
crazy on certain Martin-L\"of random sequences of a total null
measure. Such singular behavior for \emph{computable} estimators is
theoretically possible (Tomasz Steifer, private communication) but the
widely used universal codes from
\cite{ZivLempel77,Ryabko84en2,ClearyWitten84,
  KiefferYang00,CharikarOthers05,Debowski11b} decently satisfy
universality condition (\ref{Universal}) on all Martin-L\"of random
sequences since their universality rests on the Birkhoff ergodic
theorem.

\section{Three upper bounds}
\label{secBounds}

In this section, we will provide three simple upper bounds for
universal Markov orders. We will begin with the simplest one.  Namely,
the better is the code, the larger is the respective universal Markov
order. We mention this obvious behavior since it should be contrasted
with results like Theorem \ref{theoIbound} in Section~\ref{secPPM},
which conversely give the looser bound for the code-based mutual
information when the code compresses better.
\begin{theorem}
  \label{theoMOrderCompare}
  Consider universal semi-distributions 
  $\mathbf{\Pi}_1$ and $\mathbf{\Pi}_2$ such that $\mathbf{\Pi}_1(x_1^n)\ge
  \mathbf{\Pi}_2(x_1^n)$ or equivalently $\mathbf{H}_1(x_1^n)\le
  \mathbf{H}_2(x_1^n)$. Then $\mathbf{M}_1(x_1^n)\ge
  \mathbf{M}_2(x_1^n)$.
\end{theorem}
\begin{proof}
  Let $k=\mathbf{M}_1(x_1^n)$. Then
  $(n-k)h_k(x_1^n)\le \mathbf{H}_1(x_1^n)\le
  \mathbf{H}_2(x_1^n)$. Hence we obtain $\mathbf{M}_2(x_1^n)\le k$.
\end{proof}

The second bound for universal Markov orders is the pessimistic upper
bound in terms of the maximal repetition length. Let
\begin{align}
 L(x_1^n):=\max\klam{k\ge 0: \card V_k(x_1^n)< n-k+1} 
\end{align}
be the maximal repetition length \cite{DeLuca99}. Observe that
$(n-k)h_k(x_1^n)\ge \log 2=1$ and $h_k(x_1^n)\le\log D$ for
$k\le L(x_1^n)$, whereas $h_k(x_1^n)=0$ for $k>L(x_1^n)$. Hence we
have the following statement.
\begin{theorem}
  \label{theoMOrderMaxRep}
  For a universal Markov order $\mathbf{M}$,
  we have inequality
  \begin{align}
    \label{MOrderMaxRep}
    \mathbf{M}(x_1^n)\le L(x_1^n)+1.
  \end{align}
\end{theorem}
\begin{proof}
  We have $h_k(x_1^n)=0$ for $k>L(x_1^n)$, whereas
  $\mathbf{H}(x_1^n)>0$.
\end{proof}

We will improve the above naive bound in a probabilistic setting.  Let
us recall two simple bounds for the maximal repetition length. First,
since all substrings of $x_1^n$ of length $L(x_1^{n})+1$ must be
distinct we observe inequality $n-L(x_1^{n})\le D^{L(x_1^{n})+1}$,
whence we obtain the lower bound
\begin{align}
  L(x_1^{n})\ge\log_D \kwad{n-\log_D n}-1
  .
\end{align}
Moreover for any stationary ergodic measure $P$ on infinite sequences,
we have
\begin{align}
  \label{MaxRepEntropy}
  \liminf_{n\to\infty} \frac{L(X_1^n)}{\log n}\ge \frac{1}{h^P}
  \text{ almost surely}.
\end{align}
This inequality can be strict and the better bound for the left hand
side is given by the inverse (conditional) R\'enyi entropy rate
\cite{Shields92b,Szpankowski93a,Debowski18b}. In contrast, we will see
that universal Markov orders satisfy the converse inequality.

Prior to that, we will state two auxiliary statements.  First,
analogously to Theorem \ref{theoHGrowingII}, we can prove
\begin{align}
  0\le h_k(x_1^n)-\frac{n-1-k}{n-k}h_k(x_1^{n-1})\le \log
  \min\klam{2,D},
\end{align}
whence we derive
\begin{align}
  \label{HGrowing}
 h_l(x_1^{n+l})\ge \frac{n-l}{n}h_l(x_1^n).
\end{align}
This inequality can be chained with the subsequent proposition, which
says that the infinite series of some empirical entropies is upper
bounded.
\begin{theorem}
We have inequality
\begin{align}
  \label{Hbound}
  \sum_{l=0}^{\infty} h_l(x_1^{n+l})\le \log n
  .
\end{align}
\end{theorem}
\begin{proof}
Notice that
\begin{align}
  \sum_{l=0}^k h_l(x_1^{n+l})
  &=
  \sum_{w_1^{k+1}\in V_{k+1}(x_1^{n+k})}
  \frac{N(w_1^{k+1}|x_1^{n+k})}{n}
    \log \frac{n}{N(w_1^{k+1}|x_1^{n+k})}
    \le \log n.
\end{align}
Taking $k\to\infty$ yields the claim.
\end{proof}

Now we can state the improved third bound for universal Markov orders,
which is converse to inequality (\ref{MaxRepEntropy}) for the maximal
repetition length.
\begin{theorem}
  \label{theoMOrderEntropy}
  For a stationary ergodic probability measure $P$ on infinite
  sequences over a finite alphabet and a universal Markov
  order $\mathbf{M}$, we have
  \begin{align}
    \label{MOrderEntropy}
    \limsup_{n\to\infty} \frac{\mathbf{M}(X_1^n)}{\log n}
    \le
    \frac{1}{h^P}
    \text{ almost surely}.
  \end{align}
\end{theorem}
\begin{proof}
  Observe that $(n-l)h_l(x_1^n)> \mathbf{H}(x_1^n)$ for
  $l<\mathbf{M}(x_1^n)$.  Hence by formulas (\ref{HGrowing}) and
  (\ref{Hbound}), we obtain 
\begin{align}
  \log n
  &\ge
    \sum_{l=0}^{\infty} h_l(x_1^{n+l})
  \ge
  \sum_{l=0}^{\mathbf{M}(x_1^n)-1}
  \frac{(n-l)h_l(x_1^n)}{n}
  >
  \frac{\mathbf{M}(x_1^n)\mathbf{H}(x_1^n)}{n}
  .
\end{align}
In consequence we obtain an upper bound for the universal Markov
order,
\begin{align}
  \label{MlogN}
  \frac{\mathbf{M}(x_1^n)}{\log n}<\frac{n}{\mathbf{H}(x_1^n)}
  .
\end{align}
To complete the proof, we invoke the universality, i.e., property
(\ref{Universal}).
\end{proof}

As a corollary of Theorems \ref{theoConsistency} and
\ref{theoMOrderEntropy}, we obtain this proposition which asserts some
convergence of empirical entropies.
\begin{theorem}
  Consider a stationary ergodic probability measure $P$ on infinite
  sequences over a finite alphabet and a universal Markov
  order $\mathbf{M}$. Then for finite $M^P$, we have
  \begin{align}
    \label{RatesFinite}
  \lim_{n\to\infty} h_{\mathbf{M}(X_1^n)}(X_1^n)=
  \lim_{n\to\infty} h_{\mathbf{M}(X_1^n)}^P&=h^P \text{ almost surely},
\end{align}
whereas for $M^P=\infty$ and $h^P>0$, we have
\begin{align}
    \label{RatesInfinite}
  \lim_{n\to\infty} h_{\mathbf{M}(X_1^n)-1}(X_1^n)=
  \lim_{n\to\infty} h_{\mathbf{M}(X_1^n)-1}^P&=h^P \text{ almost surely}.
\end{align}
\end{theorem}
\begin{proof}
  The right equalities in (\ref{RatesFinite}) and
  (\ref{RatesInfinite}) follow by Theorem \ref{theoConsistency}. It
  remains to show the left equalities.
  
  Suppose first that $M^P=M$ is finite. Then for sufficiently large
  $n$, we have $\mathbf{M}(X_1^n)=M$ almost surely and by the Birkhoff
  ergodic theorem, we obtain almost surely
  \begin{align}
    h_M(X_1^n)\to h_M^P=h^P.
  \end{align}
  Thus we derive the left equality in (\ref{RatesFinite}).
  
  Subsequently assume that $M^P=\infty$ and $h^P>0$.  Observe that by
  (\ref{MOrderEntropy}) we have 
  $\lim_{n\to\infty} \mathbf{M}(X_1^n)/n=0$ almost surely. Then for
  any $k$ and sufficiently large $n$, we have
  $\mathbf{M}(X_1^n)-1\ge k$ and consequently almost surely
  \begin{align}
    h_{\mathbf{M}(X_1^n)-1}(X_1^n) \le
    \frac{n-k}{n-\mathbf{M}(X_1^n)+1}h_k(X_1^n)
    \to h_k^P,
  \end{align}
  whereas by the universality, i.e., property (\ref{Universal}), we
  have almost surely
  \begin{align}
   h_{\mathbf{M}(X_1^n)-1}(X_1^n)>
    \frac{\mathbf{H}(X_1^n)}{n-\mathbf{M}(X_1^n)+1}
    \to h^P.
  \end{align}
  Hence we infer the left equality in (\ref{RatesInfinite}) since
  $h^P=\inf_{k\in\mathbb{N}} h_k^P$.
\end{proof}

\section{PPM code and mutual information}
\label{secPPM}

The utility of consistent Markov order estimators exceeds the problem
of estimation of the Markov order since they can be fruitfully related
to quantification of long memory in stochastic processes, see
\cite{Talata13}.  Strengthening some results of \cite{Debowski18}, in
this section we will investigate a bound for the power-law growth of
block mutual information, which characterizes some non-hidden Markov
processes motivated by natural language phenomena
\cite{Debowski11b,Debowski12,Debowski17,Debowski18}.  This bound for
mutual information applies the concept of a universal Markov order for
the universal semi-distribution of the PPM code
\cite{Ryabko84en2,ClearyWitten84} and the respective vocabulary
size. Similar results were obtained in \cite{Debowski18} applying a
larger empirical order, equal to the Krichevsky-Trofimov order, proved
to be an inconsistent estimator of the Markov order for the uniform
measure by \cite{CsiszarShields00}. Our results have a neat
interpretation from the viewpoint of statistical language modeling,
which we will comment on at the end of this section.

Following \cite{Ryabko84en2,ClearyWitten84,Debowski18}, the PPM
measure of order $k\ge 0$ is defined as
  \begin{align}
    \PPM_k(x_1^n)&:=\prod_{i=1}^n\PPM_k(x_i|x_1^{i-1}),
  \end{align}
  where
   \begin{align}
   \label{PPMkCond}
    \PPM_k(x_i|x_1^{i-1})
                 &:=
    \begin{cases}
      \displaystyle\frac{1}{D} & \text{if $k>i-2$},
      \\
      \displaystyle\frac{N(x_{i-k}^i|x_1^{i-1})+1}{N(x_{i-k}^{i-1}|x_1^{i-2})+D}
      & \text{else}.
    \end{cases}
  \end{align}
  As we can see, $\PPM_k(x_1^n)$ is an estimator of the probability of
  block $x_1^n$ based on the Markov model of order $k$.  We note that
  these Markov estimators are adaptive, i.e., the transition
  probabilities $\PPM_k(x_i|x_1^{i-1})$ are re-estimated given each
  new symbol $x_i$. Moreover, we notice that term
  $\PPM_k(x_i|x_1^{i-1})$ is a conditional probability distribution
  and thus $\PPM_k(x_1^n)$ is a probability measure,
\begin{align}
  \sum_{x_i\in\mathbb{X}} \PPM_k(x_i|x_1^{i-1})=
  \sum_{x_1^n\in\mathbb{X}^n} \PPM_k(x_1^n)=1.
\end{align}

If $k>n-2$ then $\PPM_k(x_1^n)=D^{-n}$.  Else, by (\ref{PPMkCond}), we
obtain
\begin{align}
    \PPM_k(x_1^n)
    &=D^{-k}
    \prod_{w_1^k\in V_k(x_1^{n-1})}
    \frac{(D-1)!\prod_{w_{k+1}\in\mathbb{X}}
      N(w_1^{k+1}|x_1^n)!}{(N(w_1^k|x_1^{n-1})+D-1)!}
    .
\end{align}
Hence using the Stirling approximation, the PPM probability can be
related to the empirical entropy and the empirical vocabulary of the
respective string. In particular, by Theorem A4 in \cite{Debowski18},
we have
\begin{align}
  \alpha
  \le
    \frac{-\log\PPM_k(x_1^n)-k\log D
  -
  (n-k)h_k(x_1^n)}{D\card V_k(x_1^{n-1})}
  \le
  \log[e^2n]
  ,
  \label{PPMbound}
\end{align}
where $\alpha:=-\log (D^{-1})!$.

Subsequently, let us consider the PPM semi-distribution
\begin{align}
  \label{PPMPi}
  \mathbf{\Pi}(x_1^n)
  &:=
  \frac{6^2}{\pi^4} \cdot \frac{1}{(n+1)^2}
  \sum_{k=0}^\infty
  \frac{\PPM_k(x_1^n)}{(k+1)^{2}}
  .
\end{align}
The series can be computed effectively since $\PPM_k(x_1^n)=D^{-n}$ for
$k>n-2$.  Since the PPM semi-distribution (\ref{PPMPi}) is universal
as a consequence of bound (\ref{PPMbound}), see \cite{Ryabko84en2}, we
can consider the respective universal Markov order $\mathbf{M}(x_1^n)$
given by (\ref{MOrder}) and compare it with the Krichevsky-Trofimov
order defined as
\begin{align}
  \label{KOrder}
  \mathbf{K}(x_1^n):=
  \min\klam{k\ge 0: \PPM_k(x_1^n)=\max_{j\ge 0} \PPM_j(x_1^n)}.
\end{align}
Since $\PPM_k(x_1^n)=D^{-n}$ for $k>n-2$, the Krichevsky-Trofimov
order was shown by \cite{CsiszarShields00} to be an inconsistent
estimator of the Markov order for the uniform measure. Precisely, we
have $\lim_{n\to\infty}\mathbf{K}(X_1^n)=\infty$ almost surely for
$P(x_1^n)=D^{-n}$, which has the Markov order $M^P=0$.

This inconsistency result is quite intuitive also because the
Krichevsky-Trofimov order is greater than the universal Markov order,
which is consistent.
\begin{theorem}
  Let $\mathbf{\Pi}$ be the PPM semi-distribution (\ref{PPMPi}). We
  have
  \begin{align}
    \mathbf{M}(x_1^n)\le \mathbf{K}(x_1^n).
  \end{align}
\end{theorem}
\begin{proof}
  Let $k=\mathbf{K}(x_1^n)$. By (\ref{PPMbound}), we have
  \begin{align}
    (n-k)h_k(x_1^n)
    &< -\log\PPM_k(x_1^n)
      \nonumber\\
    &=
    -\log \max_{j\ge 0} \PPM_j(x_1^n)
    \le \mathbf{H}(x_1^n)-2\log (n+1)-\log \frac{\pi^2}{6}.
  \end{align}
  Hence $\mathbf{M}(x_1^n)\le k$.
\end{proof}

Let
$\mathbf{I}(x_1^n;x_{n+1}^m):= \mathbf{H}(x_1^n) +
\mathbf{H}(x_{n+1}^m) - \mathbf{H}(x_1^m)$ be the pointwise mutual
information for a semi-distribution $\mathbf{\Pi}$.  For the PPM
semi-distribution (\ref{PPMPi}), let
$V_\mathbf{M}(x_1^m):= V_{\mathbf{M}(x_1^m)}(x_1^m)$ and
$V_\mathbf{K}(x_1^m):= V_{\mathbf{K}(x_1^m)}(x_1^m)$ be the respective
vocabularies of the universal Markov order and the Krichevsky-Trofimov
order.  As shown in Theorem A7 in \cite{Debowski18}, stemming from
bound (\ref{PPMbound}), we have inequality
\begin{align}
    \mathbf{I}(x_1^n;x_{n+1}^m)
  \le
  \mathbf{K}(x_1^m)\log D+
  2\kwad{D\card V_\mathbf{K}(x_1^m)+2\log \frac{\pi^2}{6}+4}\log[e^2m]
  ,
\end{align}
which upper bounds the pointwise mutual information with the size of
the Krichevsky-Trofimov order vocabulary.  Applying the universal
Markov order, this inequality can be strengthened as follows.
\begin{theorem}
  \label{theoIbound}
  Let $\mathbf{\Pi}$ be the PPM semi-distribution (\ref{PPMPi}). Then for
  $0\le \mathbf{M}(x_1^m)<n,m-n<m$, we have inequality
  \begin{align}
    \mathbf{I}(x_1^n;x_{n+1}^m)
    \le
    2\kwad{D\card V_\mathbf{M}(x_1^m)+
    \frac{m\log D}{\mathbf{H}(x_1^m)}+2\log \frac{\pi^2}{6}+4}\log[e^2m]
    .
  \end{align}
\end{theorem}
\begin{proof}
  Let $k=\mathbf{M}(x_1^m)$ and $C=\frac{\pi^4}{6^2}$.  By
  (\ref{PPMPi}), we obtain
  \begin{align}
    \mathbf{H}(x_1^n)\le
    \log C+2\log (n+1)+2\log (k+1)-\log\PPM_k(x_1^n).
  \end{align}
  In consequence, by inequalities (\ref{PPMbound}) and
  (\ref{HSuperadditive}) we obtain
  \begin{align}
    \mathbf{I}(x_1^n;x_{n+1}^m)
    &=
      \mathbf{H}(x_1^n)+\mathbf{H}(x_{n+1}^m)-\mathbf{H}(x_1^m)
      \nonumber\\
    &\le 2\log C+2\log (n+1)+2\log (m-n+1)+4\log (k+1)
      \nonumber\\
    &\qquad-\log\PPM_k(x_1^n)-\log\PPM_k(x_{n+1}^m)
      -(m-k)h_k(x_1^m)
    \nonumber\\
    &\le 2\log C+2\log (n+1)+2\log (m-n+1)+4\log (k+1)+2k\log D
    \nonumber\\
    &\qquad+(n-k)h_k(x_1^n)+D\card V_k(x_1^n)\log[e^2n]
    \nonumber\\
    &\qquad+(m-n-k)h_k(x_{n+1}^m)+D\card V_k(x_{n+1}^m)\log[e^2(m-n)]
    \nonumber\\
    &\qquad-(m-k)h_k(x_1^m)
    \nonumber\\
    &\le 2\log C+4\log (m+1)+4\log (k+1)+2k\log D
    \nonumber\\
    &\qquad+2D\card V_k(x_1^m)\log[e^2m]
    \nonumber\\
    &\le 2
      \kwad{D\card V_k(x_1^m)+
      \frac{m\log D}{\mathbf{H}(x_1^m)}+\log C+4}\log[e^2m]
      .
  \end{align}
  since $k\le L(x_1^m)+1\le m$ and
  $k<\frac{m\log m}{\mathbf{H}(x_1^m)}$ by inequalities
  (\ref{MOrderMaxRep}) and (\ref{MlogN}).
\end{proof}

A similar bound holds in expectation with no restriction on $n$.
\begin{theorem}
  \label{theoIboundE}
  Let $\mathbf{\Pi}$ be the PPM semi-distribution (\ref{PPMPi}) and let
  probability measure $P$ be arbitrary. Then for any $n\ge 1$, we have
  inequality
  \begin{align}
    \sred\mathbf{I}(X_1^n;X_{n+1}^{2n})
    \le
    2\sred\kwad{D\card V_\mathbf{M}(X_1^{2n})+
    \frac{4n\log D}{\mathbf{H}(X_1^{2n})}+4\log
    \frac{\pi^2}{6}+6}\log[2e^2n]
    .
  \end{align}
\end{theorem}
\begin{proof}
  Let $C=\frac{\pi^4}{6^2}$. Since $\PPM_k(x_1^n)=D^{-n}$ for $k=n$
  then for $n\ge 1$ we obtain a uniform bound
  \begin{align}
    \label{Uniform}
    \log D\le \mathbf{H}(x_1^n)\le \log C+4\log (n+1)+n\log D.
  \end{align}
  Hence by Theorem \ref{theoIbound}, we have
  \begin{align}
    \sred\mathbf{I}(X_1^n;X_{n+1}^{2n})
    &\le
    2\sred\kwad{D\card V_\mathbf{M}(X_1^{2n})+
      \frac{2n\log D}{\mathbf{H}(X_1^{2n})}+\log C+4}\log[2e^2n]
      \nonumber\\
    &\quad+
      2\kwad{\log C+2\log (n+1)+n\log D}P(\mathbf{M}(X_1^{2n})\ge n)
      .
  \end{align}
  But by the Markov inequality and inequality (\ref{MlogN}),
  \begin{align}
    n P(\mathbf{M}(X_1^{2n})\ge n)\le \sred
    \mathbf{M}(X_1^{2n})
    \le  \sred\okra{\frac{2n}{\mathbf{H}(X_1^{2n})}}\log 2n
    .
  \end{align}
  Hence we obtain the claim.
\end{proof}

Theorems \ref{theoIbound} and \ref{theoIboundE} complement and
somewhat strengthen a series of theorems proved in \cite{Debowski18},
which allow to effectively upper bound the power-law growth of mutual
information for two strings drawn from a stationary ergodic process in
terms of the vocabulary size of a computable order.  To capture the
power-law growth of a real function $s(n)$ of natural numbers, let us
denote the Hilberg exponent
\begin{align}
  \hilberg_{n\to\infty} s(n):=\limsup_{n\to\infty}
  \frac{\log^+ s(n)}{\log n},
\end{align}
where $\log^+ x:=\log(x+1)$ for $x\ge 0$ and $\log^+ x:=0$ for $x<0$
\cite{Debowski18}.
We have, e.g., $\hilberg_{n\to\infty} n^\beta=\beta$ for $\beta\ge 0$.

Subsequently, let $\mathbf{H}^P(x_1^n):=-\log P(x_1^n)$ be the true
pointwise entropy, let $\mathbf{H}^K(x_1^n)$ be the prefix-free
Kolmogorov complexity of $x_1^n$, and let
$\mathbf{H}(x_1^n):=-\log \mathbf{\Pi}(x_1^n)$ for the PPM semi-distribution
(\ref{PPMPi}). Then for any stationary measure $P$, we have
\begin{align}
  &\hilberg_{n\to\infty} \kwad{\sred\mathbf{H}^P(X_1^n)-nh^P}
    =\hilberg_{n\to\infty} \sred\mathbf{I}^P(X_1^n;X_{n+1}^{2n})  
    \nonumber\\
  &\le
    \hilberg_{n\to\infty} \kwad{\sred\mathbf{H}^K(X_1^n)-nh^P}
    =\hilberg_{n\to\infty} \sred\mathbf{I}^K(X_1^n;X_{n+1}^{2n})
    \nonumber\\
  &\le
    \hilberg_{n\to\infty} \kwad{\sred\mathbf{H}(X_1^n)-nh^P}
    \le
    \hilberg_{n\to\infty} \sred\mathbf{I}(X_1^n;X_{n+1}^{2n})
    ,
    \label{MIRedundancy}
\end{align}
as shown in Theorem A1 in \cite{Debowski18}.

In the following, we can strengthen Theorem A8 of \cite{Debowski18},
which asserts that
\begin{align}
  \hilberg_{n\to\infty} \sred\mathbf{I}(X_1^n;X_{n+1}^{2n})
    \le \hilberg_{n\to\infty} \sred
      \kwad{\mathbf{K}(X_1^n)+\card V_\mathbf{K}(X_1^n)}
      .
\end{align}
Our strengthened proposition applies the universal Markov order
vocabulary.
\begin{theorem}
  \label{theoMIWords}
  Consider a stationary probability measure $P$ on infinite sequences
  over a finite alphabet.  Let $\mathbf{\Pi}$ be the PPM semi-distribution
  (\ref{PPMPi}). Denote the R\'enyi block entropy
  $R^P_n:=-\log\sred P(X_1^n)$.  If the R\'enyi entropy rate is
  strictly positive, i.e., $\inf_{n\ge 1} R^P_n/n>0$ then
  \begin{align}
    \hilberg_{n\to\infty} \sred\mathbf{I}(X_1^n;X_{n+1}^{2n})
    \le \hilberg_{n\to\infty} \sred
      \card V_\mathbf{M}(X_1^n)
      .
    \label{MIWords}
 \end{align}
\end{theorem}
\begin{proof}
  The inequality follows by Theorem \ref{theoIboundE} and the lower
  bound in (\ref{Uniform}) since using the Markov inequality we can
  further bound
\begin{align}
  \sred\okra{\frac{1}{\mathbf{H}(X_1^n)}}
  &=
    \int_0^\infty P\okra{\frac{1}{\mathbf{H}(X_1^n)}\ge p}dp
    =
    \int_{\log D}^\infty P\okra{\mathbf{H}(X_1^n)\le u}\frac{du}{u^2}
    \nonumber\\
  &\le
    \int_{\log D}^\infty \min_\alpha
    \kwad{
    P\okra{\frac{2^{-\mathbf{H}(X_1^n)}}{P(X_1^n)2^\alpha}\ge 1}
    +
    P\okra{\frac{P(X_1^n)2^\alpha}{2^{-u}}\ge 1}  
    }\frac{du}{u^2}
    \nonumber\\
  &\le
    \int_{\log D}^\infty \min_\alpha
    \kwad{
    2^{-\alpha}
    +
    \sred P(X_1^n)2^{u+\alpha}
    }\wedge 2\frac{du}{u^2}
    \nonumber\\
  &\le
    2\int_{\log D}^\infty \kwad{\sred P(X_1^n)2^{u}}^{1/2}\wedge 1\frac{du}{u^2}
    \nonumber\\
  &\le
    2\int_0^{R^P_n}\frac{du}{(R^P_n)^2}
    +
    2\int_{R^P_n}^\infty\frac{du}{u^2}
   \le
    \frac{4}{R^P_n}
    \label{InverseH}
\end{align}
if $f(R^P_n)\ge f(\log D)$ for the convex function $f(u):=2^{u/2}/u^2$.
\end{proof}

Compared to the results of \cite{Debowski18} applying the vocabulary
size of the Krichevsky-Trofimov order, Theorems \ref{theoIbound}
through \ref{theoMIWords} allow to take the vocabulary size of a
smaller order, equal to the universal Markov order. This smaller order
enjoys a neat interpretation of a consistent Markov order estimator
and is usually negligible compared to the vocabulary size because of
the upper bound (\ref{InverseH})---in contrast to the unknown
worst-case behavior of the Krichevsky-Trofimov order.  Moreover,
according to Theorem \ref{theoMIWords} and inequalities
(\ref{MIRedundancy}), if the true mutual information for a stochastic
source grows like a power law, which holds for certain
non-finite-state processes
\cite{Debowski11b,Debowski12,Debowski17,Debowski18}, then the
vocabulary size of the universal Markov order must also grow like a
power law. Analogous results were obtained not only for the PPM
semi-distribution in \cite{Debowski18} but, earlier, also for minimal
grammar-based codes in \cite{Debowski11b}, which unfortunately seem to
be computationally intractable \cite{CharikarOthers05}.

Let us also add that there is a hypothesis by Hilberg \cite{Hilberg90}
that the mutual information for natural language grows like a
power-law indeed, which can be further supported not only by large
scale computational experiments
\cite{TakahiraOthers16,HahnFutrell19,BravermanOther19,KaplanOther20}
but also by a reasoning linking semantics and non-ergodicity. This
semantic interpretation of non-ergodicity provides a lower bound for
the mutual information and, combined with propositions such as Theorem
\ref{theoMIWords}, it yields an intriguing claim that the number of
independent timeless facts described by a random text must be roughly
smaller than the number of distinct words in the text, see
\cite{Debowski11b,Debowski18} for the technical details. Since the
number of distinct words grows roughly like a power of the text size
for natural language, which is called Herdan's or Heaps' law
\cite{Herdan64,Heaps78}, one may hope that the number of independent
timeless facts described by texts in natural language is also
large. In particular, natural language cannot be not a finite-state
process.

\section*{Acknowledgment}

We thank Boris Ryabko for pointing out to us his very relevant
references, which we failed to cite in the first draft of this paper.
We apologize for this unintentional omission.

\bibliographystyle{IEEEtran}

\bibliography{0-journals-abbrv,0-publishers-abbrv,ai,mine,tcs,ql,books,nlp}

\end{document}